\pdfoutput=1
\documentclass[11pt,letterpaper]{article}
\usepackage{fullpage}
\usepackage{lmodern}
\usepackage[T1]{fontenc}
\usepackage[utf8]{inputenc}
\usepackage[tbtags]{amsmath}
\usepackage{amssymb, amsthm}
\usepackage{mathtools}
\usepackage{hyperref}
\usepackage[capitalise,nameinlink]{cleveref}
\hypersetup{colorlinks={true},linkcolor={blue},citecolor=red}
\usepackage{comment}
\usepackage{algorithm2e}
\usepackage{todonotes}

\newtheorem{theorem}{Theorem}
\newtheorem{lemma}{Lemma}
\newtheorem{proposition}{Proposition}

\theoremstyle{definition}

\newcommand{\sset}[1]{\left\{ #1\right\}}

\newcommand{\fwh}[1]{\; \left| \; #1 \right.}
\newcommand{\fwhs}[1]{\; | \; #1 }
\newcommand{\opt}{\ensuremath{\mathrm{OPT}}}

\newcommand{\expectation}{\mathbb{E}}
\newcommand{\expect}[1]{\ensuremath{\expectation\left[#1\right]}}
\newcommand{\prob}[1]{\ensuremath{\mathrm{Pr}\left[#1\right]}}

\date{\empty}

\begin{document}

\title{Online trading as a secretary problem\thanks{Supported by the ERC Advanced Grant 321171 (ALGAME).}}

\author{
Elias Koutsoupias\thanks{Department of Computer Science, University of 
Oxford. Email: \texttt{\href{mailto:elias@cs.ox.ac.uk}{elias@cs.ox.ac.uk}}}
  \and
Philip Lazos\thanks{Department of Computer, Control, and Management 
Engineering ``Antonio Ruberti,'' Sapienza University of Rome, 
Email: \texttt{\href{mailto:lazos@diag.uniroma1.it}{lazos@diag.uniroma1.it}}}
}

\maketitle

\begin{abstract}
  We consider the online problem in which an intermediary trades identical items with a sequence of $n$ buyers and $n$ sellers, each of unit demand. We assume that the values of the traders are selected by an adversary and the sequence is randomly permuted. We give competitive algorithms for two objectives: welfare and gain-from-trade.
\end{abstract}

\section{Introduction}
We study the problem of facilitating trade between $n$ buyers and $n$ 
sellers that arrive online. We consider one of the simplest settings in which 
each trader, buyer or seller, is interested in trading a single item, and all 
items are identical. Each trader has a value for the item; a seller will sell to 
any price higher than its value and a buyer will buy for any price lower than 
its value. Upon encountering a trader, the online algorithm makes an 
irrevocable price offer, the trader reveals its value and, if the value is at the 
correct side of the offered price the item is traded. After buying an item 
from a seller, the online algorithm can store it indefinitely to sell it to later 
buyers. Of course, the online algorithm can only sell to a buyer if it has at 
least one item at the time of the encounter.

We consider online algorithms that offer prices based on the sequence of past values and we assume that the online algorithm knows only the number of buyers and sellers, but not their values. The values of the sellers and buyers are selected \emph{adversarially and are randomly permuted}. In that respect, the problem is a generalization of the well-known secretary problem. The secretary problem corresponds to the special case in which there are only buyers, the algorithm starts with a single item, and the objective is to maximize the total welfare, which is to give the value to a buyer with as high value as possible.

Extending this to both sellers and buyers, creates a substantially richer setting. One of the most important differences between the two settings is that besides the objective of maximizing the total welfare, we now have the objective of maximizing the gain-from-trade. For both objectives, the algorithm must buy from sellers with low values and sell to buyers with high values. The objective is that at the end, the items end up at the hands of the traders, sellers or buyers, with the highest values. The welfare of a solution is defined as the value of the buyers and sellers that have an item. The gain-from-trade of a solution is the difference between the welfare at the end of the process minus the welfare at the beginning. At optimality the two objectives are interchangeable: an algorithm achieves the maximum welfare if and only if it achieves the maximum gain-from-trade. But for approximate solutions, the two objectives are entirely different, with the gain-from-trade being the most demanding one.

The Bayesian version of the problem, in which the values of the buyers and sellers are drawn from known probability distributions has been extensively considered in the literature. Optimal mechanisms for bilateral trading, that is, the offline case of a single seller and a single buyer, were first analysed by Myerson and Satterthwaite in \cite{myerson_bilateral} and played a pivotal role in the development of the area (see the section Related Work). The online Bayesian case was considered in \cite{gkl2017}, where the values are drawn from a known distribution but the sequence is adversarially ordered.

A generalization of our model is when the items are not identical and each buyer has different value for each one of them, i.e., each seller has a value for its item and each buyer has a vector of values, one for every pair buyer-seller. This is also a generalization of the well-studied online maximum-matching problem~\cite{Korula2009,kesselheim}. One can cast the online maximum-matching problem as the version in which the sellers arrive first and have zero value for their item. The optimal online algorithm for this problem has competitive ratio $1/e$, when the objective is the welfare (which in the absense of seller values is identical to the gain-from-trade). Our model is incomparable to the online maximum-matching problem: it is simpler in the sense that the items are identical (a single value for each buyer instead of a vector of buyer-item values), and at the same time more complicated in that the items are not present throughout the process, but they are brought to the market by sellers that have their own utility. The fact that in our model the buyer-item values are related, allows for a much better competitive ratio regarding the welfare, (almost) $1$ instead of $1/e$. More importantly, our algorithm is truthful, while in contrast, no good truthful algorithm is known for the online maximum-matching problem, which remains one of the main open problems of the area. On the other hand, the introduction of sellers poses new challenges, especially with respect to the objective of the gain-from-trade.

There are also similarities between our model and the extension of the classical secretary problem to $k$ secretaries. From an influential result by Kleinberg \cite{Kleinberg_kse} we know that this problem has competitive ratio $1-1/\sqrt{k}$ which is asymptotically tight, and can be transformed into a truthful algorithm. This result depends strongly on the knowledge of $k$. In our case the equivalent measure, the \emph{number of trades} is not known from the beginning and has to be learned, with a degree of precision that is crucial, especially for the gain-from-trade objective. The fact that the gain-from-trade is not monotone as a function of time highlights the qualitative difference between the two models; the gain-from-trade temporarily \emph{decreases} when the algorithm buys an item, with the risk of having a negative gain at the end. More generally, with the mix of buyers and sellers, wrong decisions are penalized more harshly and the monotone structure of the problem is disrupted.

\subsection{Our results}

We consider the case when both the number of buyers and the number of sellers is $n$. For welfare we show a competitive ratio of $1-\tilde O(n^{-1/3})$, where $\tilde O$ hides logarithmic factors.

Actually we can compare an online algorithm with two offline benchmarks: the \emph{optimal} benchmark, in which all trades between buyers and sellers are possible, independently of their order of appearance, and the expected \emph{sequential optimal} in which an item can be transferred from a seller to a buyer only if the seller precedes the buyer in the order.

Our online algorithm achieves a competitive ratio of $1-\tilde O(n^{-1/3})$ 
against the optimal benchmark. To achieve this, it has a small sampling 
phase of length $\tilde{O}(n^{2/3})$ to estimate the \emph{median} of the 
values of all traders, and then uses it as a price for the remaining traders. 
But if the optimal number of trades is small, such a scheme will fail to 
achieve competitive ratio almost one, because with constant probability 
there will not have enough items to sell to buyers with high value. To deal 
with this risk, the algorithm not only samples values at the beginning but it 
additionally buys sufficiently many items, $\tilde{O}(n^{2/3})$, from the 
first sellers\footnote{Buying from the first sellers cannot be done truthfully 
unless the algorithm knows an upper bound on their value. But this is not 
necessary since there is an alternative that has minor effects on the 
competitive ratio: the algorithm offers each seller the maximum value of 
the sellers so far. This is a truthful scheme that buys from all but a 
logarithmic number of sellers, in expectation.}. The number 
$\tilde{O}(n^{2/3})$ of bought items balances the potential loss of the 
welfare that results from removing items from sellers to the expected loss 
from not having enough items for buyers of high values.

The term $O(n^{-1/3})$ in the competitive ratio seems to be optimal for a 
scheme that fixes the price after the sampling phase and relates to the 
number of items needed to approximate the median to a good degree. It 
may be possible to improve this term to $O(n^{-1/2})$ by a more adaptive 
scheme, as in the case of the $k$-secretary problem \cite{Kleinberg_kse}. 
Finally, it may be possible to remove the logarithmic factors from the 
competitive ratio, but we have opted for simplicity and completeness.

For the objective of gain-from-trade, we give a truthful algorithm that has a constant competitive ratio, assuming that the algorithm starts with an item. The competitive ratio is high, approximately $10^3$, but it drops to a small constant when the optimal number of trades is sufficiently high. The additional assumption of starting with an item is necessary, because without it, no online algorithm can achieve a bounded competitive ratio.

The main difficulty of designing an online algorithm for gain-from-trade is that even a single item that is left unsold at the end has dramatic effects on the gain-from-trade. The online algorithm must deal with the case of many traders, large welfare, but few optimal trades and small gain-from-trade.

To address this problem, our algorithm, unlike the case of welfare, has a large sampling phase. It uses this phase to estimate the number of optimal trades and two prices for trading with buyers and sellers. If the expected number of optimal trades is high, the algorithm uses the two prices for trading with the remaining traders. But if the number is small, it runs the secretary algorithm with the item that it starts with.

The analysis needs high concentration bounds on the expected number of trades to minimize the risk of having items left unsold. Our algorithm is ordinal, in the sense that it uses only the order statistics of the values not the actual values themselves. This leaves little space for errors and it may be possible that cardinal algorithms that use the actual values can do substantially better.

\section{Related Work}

% On a technical level, our work is more closely related to the secretary 
% problem and online matching literature, but borrows the setting and the 
% aspects of truthfulness and rationality from mechanism design, specifically 
% from the bilateral trade setting.

The bilateral trade literature was initiated by 
Myerson and Satterthwaite in their seminal paper \cite{myerson_bilateral}. 
They investigated the case of a single 
seller-buyer pair and proved their famous impossibility result: there exists 
no truthful, individually rational and budget balanced mechanism that also 
maximizes the welfare (and consequently, the gain from trade).
Subsequent research studied how well these two objectives can be 
approximated by relaxing these conditions. Blumrosen and 
Mizrahi~\cite{blumrosenGFT_WINE} devised a $1/e$-approximate, Bayesian 
incentive compatible mechanism for the gain from trade assuming the 
buyer's valuation is monotone hazard rate. Brustle et al. expanded in this 
direction in~\cite{brustle2017approximating} for arbitrary valuations and 
downwards closed feasibility constraints over the allocations. In the case 
where there are 
multiple, unit demand, buyers and sellers, McAfee provided a weakly 
budget balanced, $1-1/k$ approximate mechanism for the gain from trade 
in~\cite{mcafee_dominant_1992}, where $k$ is the number of trades in the 
optimal allocation. This was later extended to be strongly budget balanced 
by  Segal-Halevi et al. in \cite{segal2016sbba}. McAfee also proved a simple 
$2$-approximation to the gain from trade if the buyer's median valuation is 
above the seller's \cite{mcafee_gains_2008}. This was significantly 
improved by 
Colini-Baldeschi et al. in~\cite{leonardiGFT} to $1/r$ and $O(\log(1/r))$, 
where $r$ is the probability that the buyer's valuation for the item is higher 
than the seller's. Recently, Giannakopoulos et al.~\cite{gkl2017} studied an 
online variant 
of this setting where buyers and sellers are revealed sequentially by an 
adversary and have known prior distributions on the value of the items.

The random order model we are using has its origins in the well-known 
secretary problem, where $n$ items arrive in online fashion 
and our goal is to maximize the probability of selecting the most valuable, 
without knowing their values in advance. The matroid secretary problem 
was 
introduced by Babaioff et al.~\cite{Babaioff_matroid_sec}. In this setting, 
we are allowed to select more than item, provided our final selection 
satisfies matroid constraints. A variety of different matroids have been 
studied, with many recent results presented by Dinitz in \cite{Dinitz:2013}. 
Of particular interest to our problem are secretary problems on bipartite 
graphs. Here, the left hand side vertices of the graph are fixed and the right 
hand side vertices (along with their incident) edges appear online. The 
selected edges must form a (incomplete) matching and the goal is to 
maximize the sum of their weights. Babaioff 
et al. in \cite{Babaioff_matroid_sec} provided a 
$4d$-competitive algorithm for the transversal matroid with bounded left 
degree $d$, which is a special case of the online bipartite matching where 
all edges connected to the same left hand side vertex have equal value. This 
was later improved to $16$ by Dimitrov and Plaxton~\cite{Dimitrov2012}. 
The case where all edges have unrelated weights was first considered by 
Korula and Pal in~\cite{Korula2009} who designed a $8$-competitive 
algorithm, which was later improved to the optimal $1/e$ by Kesselheim et 
al. \cite{kesselheim}. Another secretary variant which is close to our work is 
when the online selects $k$ items instead of one, where 
Kleinberg~\cite{Kleinberg_kse} showed 
an asymptotically tight algorithm with competitive ratio $1 - O(\sqrt{1/k})$.

The wide range of applications of secretary models (and the related prophet 
inequalities) have led to the design of posted price mechanisms, that are 
simple to describe, robust, truthful and achieve surprisingly good 
approximation ratios. Hajiaghayi et al. introduced prophet inequality 
techniques in online auction in \cite{hajiaghayi_automated_2007}. The 
$k$-choice secretary described above was then studied in 
\cite{hajiaghayi_adaptive_2004} which combined with 
\cite{Kleinberg_kse} yielded an asympotically optimal, truthful mechanism.
For more general auction settings, posted-price mechanisms have been 
used by Chawla et al. in \cite{chawla_multi-parameter_2010} for unit 
demand agents and expanded by Feldman et al. in \cite{Feldman2015} for 
combinatorial auctions and \cite{feldman} for online budgeted settings.

\section{Model and Notation}
The setting of the \emph{random intermediation} problem consists of 
sets 
$B = \{b_1,\ldots,b_n\}$ and $S = \{s_1,\ldots, s_n\}$ containing the 
valuations of the buyers and sellers. For convenience, we assume that they 
are all distinct.
The intermediary interacts with a uniformly 
random permutation $\sigma$ of $B \cup S$ which is presented to him one 
agent at a time, over $2n$ steps.  The intermediary has no knowledge of 
$\sigma(t)$ before step $t$. We use $b^i$ and $s^j$ to denote the $i$-th 
\emph{highest} valued seller and $j$-th \emph{lowest} valued seller 
respectively.

We study \emph{posted price} mechanisms that upon seeing the identity of 
agent $t$ offer price $p_t$. This price can not depend on
the entire valuation function; only the values within 
$\sigma(1)\ldots\sigma(t-1)$ which
are revealed at this point. We buy or sell one item from sellers or buyers 
who accept our price, respectively. Of course, we can only sell items if we 
have stock available.
Formally, let $\kappa_t$
be the number of items at time $t$. Starting with $\kappa_0$ items (with 
$\kappa_0=0$ for welfare and $1$ for the gain-from-trade):
\begin{equation*}
\kappa_{t+1}=
\begin{cases}
\kappa_{t}+1, 
	&\text{if}\;\; \sigma(t)\in S \; \land \; \sigma(t)\leq p_{t}\\
\kappa_{t}-1,
	&\text{if}\;\; \sigma(t)\in B \; \land \; \kappa_{t}\ge 1 \; \land \; 
	\sigma(t)\geq p_{t}\\
\kappa_{t}
	&\text{otherwise}.
\end{cases}
\end{equation*}
The set of sellers from whom we bought items during the algorithm's 
execution is $$T_S=\sset{s \in S \fwh{ \exists t~ \sigma(t) = s \leq 
p_t}}$$ and similarly the set of 
buyers we sold to is $T_B=\sset{b \in B \fwh{\exists t ~\sigma(t) = 
	b \geq p_t \land \kappa_t > 0}}$.
Notice that these are random variables, depending on $\sigma$.

The social \emph{welfare} of online algorithm $A$ is the sum of the 
valuations of all agents \emph{with} items. In particular, after executing 
$A$ it is: $\mathcal{W}_A(S,B)= \expect{\sum_{s\in S\setminus T_S} s + \sum_{b \in T_B} b}$.
The gain from trade (or GFT) produced by algorithm $A$ throughout the run 
is the 
difference between the final and starting welfare: $ 
GFT_A(S,B)=\expect{\sum_{b\in T_B} b - \sum_{s\in T_S} s}$.

We are interested in the \emph{competitive ratio} of our online algorithm 
$A$ compared to the offline algorithm $OPT$. In this setting there are two 
different offline algorithms to compare against: optimal offline and 
sequential offline. They both know $S,B$, but the first can always achieve 
the maximum welfare, whereas the second operates under the same 
constrains as we, namely he can only perform trades permitted by 
$\sigma$, which is unknown. We say that 
algorithm $A$ is $\rho$-competitive 
for welfare (or gain from trade) if for any $B,S$ we have:
\begin{equation}
\mathcal{W}_A(B,S) \ge \rho \cdot \mathcal{W}_\opt(S,B)-\alpha,
\end{equation}
for some fixed $\alpha \ge 0$. 

Often we will refer to the \emph{matching} between a set of buyers and a 
set of sellers. Let $M(S,B) = \sset{ \{S_1\} \cup \{B_1\}}$, where $S_1 
\subseteq S, B_1 \subseteq B$ is the set of sellers and buyers with whom 
we trade (or are matched, in the sense that the items move from sellers to 
buyers) in a welfare maximizing allocation and $m(S,B)$ the optimal gain 
from trade. Note that this does \emph{not} contain pairs: only the 
set of each side of the matching. Similarly, let $M(S,B,q,p)$ be the matching 
generated by only trading with 
sellers valued below $q$ and buyers above $p$. In a slight abuse of 
notation, we will use $|M(S,B)| = |S_1|$ for the size of the matching and 
$M(S,B) \cap M(S',B') = \{ \{S_1 \cap S_1'\} \cup \{B_1 \cap B_1'\}\}$.

\section{Welfare}

In order to approximate the welfare, the online algorithm uses a
sampling phase to find the median price, in an attempt to transfer
items from agents below the median to more valuable ones above it. The
two main challenges, in terms of its performance, are estimating the
median with a small sample and not missing too many trades due to the
online nature of the input.  Before we delve into the actual
algorithm, it is useful to state two probability concentration
results, similar to the familiar Azuma-Hoeffding inequality, but for
the setting where sampling happens \emph{without} replacement as is
our case.

\begin{lemma} \label{lemma:stronger-concentration}
	Let $\mathcal{X} = \{x_1,\ldots,x_N\}$ where $x_i \in \sset{0,1}$,
	$x_1 = x_2 = \ldots = x_m = 1$ and $x_{m+1} = \ldots = x_N = 0$ for some 
	integer $m \ge 0$. Consider 
	sampling $n$ values of 
	$\mathcal{X}$ uniformly at random \textbf{without} replacement and let 
	$X_i$ be the value of the $i-th$ draw. For $Y = \sum_{i=1}^n X_i$, we 
	have that for any $\epsilon > 0$:
	\begin{equation}
	\Pr[Y \ge (1+\epsilon) \expect{Y}] \le e^{-2\epsilon^2 
		\max\{m,n\}
		\frac{mn}{N^2}}
	\end{equation}
	and
	\begin{equation}
	\Pr[Y \le (1-\epsilon) \expect{Y}] \le e^{-2\epsilon^2 
		\max\{m,n\}
		\frac{mn}{N^2}}.
	\end{equation}
\end{lemma}
\begin{proof}
	Let $Y_i = \expect{Y | X_1,\ldots, X_i}$ be the  
	Doob martingale of $Y$, exposing the choices of the first $i$ draws. 
	Clearly we 
	have that $|Y_{i+1} - Y_i| \le 1$, since the knowledge of one draw 
	cannot change the expectation by more than $1$. Applying Azuma's 
	inequality, we obtain:
	\begin{equation}\label{eqn:conc_1}
	\Pr[Y_n - Y_0 \ge t] \le e^{\frac{-t^2}{n}}.
	\end{equation}
	
	Let $Z_j$ for $1\le j \le m$ indicate if $x_j$ was chosen. Since only 
	these $x_j$ contribute to $Y$, we have that $Y = \sum_{i=1}^{m} Z_i$.
	Repeating the previous martingale construction, we get:
	\begin{equation}\label{eqn:conc_2}
	\Pr[Y_m - Y_0 \ge t] \le e^{\frac{-t^2}{m}}.
	\end{equation}
	
	But, we know that $Y_0 = \expect{Y} = \expect{\sum_{i=1}^m Z_i} = m 
	\frac{n}{N}$. Setting $t = \epsilon m \frac{n}{N}$ in both 
	\eqref{eqn:conc_1} and \eqref{eqn:conc_2} and using $Y_n = Y_m = 
	Y$ we obtain:
	\begin{equation}
	\Pr[Y \ge (1+\epsilon) \expect{Y}] \le e^{-2\epsilon^2 
		\max\{m,n\}
		\frac{mn}{N^2}}.
	\end{equation}
	Concentration in the opposite direction is found by repeating the same 
	analysis, using the complementary form of Azuma's inequality.
\end{proof}
Note that this result is not superfluous: by immediately applying 
Hoeffding's inequality for sampling with 
replacement, we would obtain:
\begin{equation*}
\Pr[Y \ge (1+\epsilon) \expect{Y}] \le e^{-2\epsilon^2 
	\frac{m^2n}{N^2}},
\end{equation*}
which is only tight if $m$ is large \emph{compared to} $N$. The concentration 
should intuitively work if $n$ is a large fraction of $N$ as well: imagine $n=N$.  \\

Similarly, we often encounter a situation where we are interested in the 
number of trades between $n$ sellers and $n$ buyers, arriving in a 
uniformly random permutation. Assuming we buy from all sellers, 
occasionally we would encounter a buyer without having any items at hand.
This results shows that even though this is the case, few trades are lost.

\begin{lemma}\label{lemma:matching-size}
	The number of trades $M(\sigma)$, where $\sigma$ is a uniformly random 
	sequence containing $n$ buyers and $n$ sellers, is:
	\begin{equation}
	\expect{M(\sigma)} \ge n - \sqrt{4n \log n}\left(1 - \frac{1}{n}\right) - 2 - 
	\sqrt{2n} \ge n - \sqrt{8n\log n}
	\end{equation}
	assuming all sellers are valued below all buyers.
\end{lemma}
\begin{proof}
	To calculate the number of trades $M(\sigma)$, we will compare against a 
	simpler sequence, $\hat \sigma$, which also has length $2n$ but at each step 
	a fair coin is flipped, determining if the next agent is a buyer or seller. 
	Intuitively 
	$\hat \sigma$ is like $\sigma$ with replacement and we expect it to behave 
	similarly, since the discrepancy between buyers and sellers is asymptotically 
	insignificant compared to their expected number.  
	
	\newcommand{\discr}{\ensuremath{D}}
	Let $\discr$ be the absolute difference between buyers and sellers in $\hat 
	\sigma$ and $U(\sigma)$ the number of unsold items after $\sigma$. Then:
	\begin{align}
	\expect{U(\hat \sigma)} 
	&= \sum_{i=0}^{2n} \expect{U(\hat \sigma) \fwhs{\discr = i}} \notag
	\cdot \prob{\discr = i}\\
	&\ge \sum_{i=0}^{2n} \left(\expect{U(\hat \sigma) \fwhs{\discr = 0}} - i\right)
	\cdot \prob{\discr = i} \label{ineq:all-sellers} \\
	&= \sum_{i=0}^{2n} \expect{U(\sigma)} \cdot \prob{\discr = i} 
	+ \sum_{i=0}^{2n} i \cdot \prob{\discr = i} \label{eq:discr-zero} \\
	&= \expect{U(\sigma)}  - \expect{D} \notag,
	\end{align}
	where in \Cref{ineq:all-sellers} we used that $\discr=i$ can imply at most 
	$i$ extra buyers, which means at most $i$ fewer unsold items and in 
	\Cref{eq:discr-zero} we used that $\discr = 0$ implies that $\hat \sigma$ 
	contains 
	exactly $n$ buyers and $n$ sellers and is in fact distributed exactly like 
	$\sigma$.
	
	The difference between buyers and sellers is distributed exactly like the 
	distance 
	from the origin of a one-dimensional, balanced random walk with $2n$ steps, 
	which is known to have $\expect{D} \le \sqrt{2n}$. For the remaining term we 
	need the following lemma.
	
	\begin{lemma}
		$$\expect{U(\hat \sigma)} \le \sqrt{4n \log n}\left(1 - \frac{1}{n}\right) 
		+ 
		2$$
	\end{lemma}
	\begin{proof}
		Let $Z_i$ denote the change in unsold items right after step $i$. Clearly, 
		$\sum_{i=1}^{2n} Z_i = U(\hat \sigma)$
		The process $Z_i$ is almost a martingale but not quite:  clearly 
		$\expect{Z_i} \le n$ for all $i$ and we do have 
		$\expect{Z_{i+1}|Z_i \ge 1} = Z_i$ since it is always equally likely to 
		encounter 
		a seller or a buyer. However, $\expect{Z_{i+1}|Z_i = 0} > 
		Z_i$, because if we encounter a buyer we won't be able to sell.
		
		We can define $Y_{i}$ in the same probability space, where $Y_0=0$, 
		and
		\begin{equation}
		Y_{i+1} = Y_i +  
		\begin{cases}
		Z_{i+1} &\text{if } Y_i > 0\\
		-Z_{i+1} &\text{if } Y_i < 0\\
		\begin{cases}
		Z_{i+1} &\text{with probability } \frac{1}{2}\\
		-Z_{i+1} &\text{with probability } \frac{1}{2}
		\end{cases} &\text{if } Y_i = 0
		\end{cases}.
		\end{equation}
		The crucial observation is that $Y_i$ has no barrier at 0. Notice, that $|Y_i| 
		= 
		\sum_{j=1}^i Z_i$ 
		for all $i$
		and $Y_i$ is a martingale.
		
		Moreover, we have that $|Y_{i+1} - Y_i| \le 1$ thus by the 
		Azuma-Hoeffding inequality we can bound the expected value
		$\expect{Y_{2n}}$:
		\begin{align}\label{eqn:bound_unsold}
		\Pr[|Y_{2n}| \ge x] = \Pr[|Y_{2n} - Y_0| \ge x] \le 
		2e^{\frac{-x^2}{2\cdot 2n}} \Rightarrow\\
		\expect{U(\hat \sigma)} &\le x\left(1 - 2e^{\frac{-x^2}{4n}}\right) + 
		2n e^{\frac{-x^2}{4n}},
		\end{align}
		where we also used that $U(\hat \sigma) \le 2n$ by definition.
		We can set $x = \sqrt{4n \log n}$ to obtain the simpler 
		form:
		\begin{equation}\label{ineq:bound_unsold}
		\expect{U(\hat \sigma)} \le \sqrt{4n \log n}\left(1 - \frac{1}{n}\right) 
		+ 2.
		\end{equation}
	\end{proof}
	
	Putting everything together, we have that:
	\begin{align*}
	\expect{M(\sigma)} = n - U(\sigma)
	&\ge n -  \expect{M(\hat \sigma)} - \expect{\discr}\\
	&\ge n - \sqrt{4n \log n}\left(1 - \frac{1}{n}\right) - 2 - \sqrt{2n}
	\end{align*}
\end{proof}

All the machinery is now in place analyse sequential algorithms in this 
setting. We first show a key property of the offline algorithm.
\begin{proposition}
	The optimal offline algorithm sets a price $p$, equal to the 
	median of all 
	the agents' valuations and trades items from sellers valued below $p$ to 
	buyers valued above $p$.
\end{proposition}
\begin{proof}
	Since there are only $n$ items available, if we could freely redistribute 
	the items we would choose the top $n$ agents with highest valuations. 
	Let $p$ be the value of the $n$-th most valuable agent. If there are $k$ 
	buyers valued above $p$ we have $n-k$ buyers and $k$ sellers valued 
	below it. Thus, buying from all sellers below $p$ and selling to all buyers 
	above it is an optimal algorithm.
\end{proof}

However, the optimal sequential offline algorithm would not just trade at 
this price. 
For instance, if there is $1$ buyer and $n-1$ sellers above $p$ and $1$ 
seller and $n-1$ buyers below,
trading at this price would give a $1/2$ probability of transferring the item, 
since only one transfer increases the welfare and the agents have to appear 
in the right order.
Therefore, if that buyer has a much larger valuation than anyone else, this 
algorithm would only be $1/2$-competitive. However, we can modify this 
approach with a bias towards buying more items than needed, in order to 
maximise the probability of finding high valued buyers. 
\begin{lemma} \label{lemma:sequential}
	The optimal sequential online algorithm is $\left(1-O\left(\frac{\log 
		n}{n^{1/3}}\right)\right)$-competitive against the optimal 
	offline for welfare.
\end{lemma}
\begin{proof}
	The optimal online algorithm adjusts the price $p$ according to the 
	following two cases, where $M = |M(S,B)|$.
	\begin{enumerate}
		\item $M \ge n^{2/3}$.
		In this case the same price $p$ is used. At the end, the online 
		algorithm will still keep the highest valued $n - M$ sellers and by 
		Lemma~\ref{lemma:matching-size} will match at least
		\begin{equation*}
		M - \sqrt{8M\log M}
		%\frac{M-1}{M}\left(M - \sqrt{2M \log M}\right)
		\end{equation*}
		buyers in expectation. The offline optimum will of course keep the 
		highest $n-M$ sellers and $M$ buyers, leading to a competitive ratio 
		of at most:
		\begin{equation*}
		\frac{M - \sqrt{8M\log M}}{M} = 1 - \frac{\sqrt{8\log M}}{M^{1/3}} = 1 - 
		O\left(\frac{\log n}{n^{1/3}}\right).
		%\frac{M-1}{M}\left(1 - \frac{\sqrt{2\log M}}{\sqrt{M}}\right) = 
		%1 - O\left(\frac{\log n}{n^{1/3}}\right).
		\end{equation*}
		\item $M < n^{2/3}$.
		
		In this case, suppose two prices are used: $p_S$ to buy from the 
		lowest $n^{2/3}$ sellers and $p_B$ to sell to the highest $n^{2/3}$ 
		buyers.
		For the buyers, the online does at least as well as the previous case. 
		In particular, it it obtains a uniformly random sample of size at least 
		$n^{2/3} - \sqrt{8n^{2/3} \log n^{2/3}}$ (in expectation) by 
		Lemma~\ref{lemma:matching-size}, amongst the top $n^{2/3}$ 
		buyers.  Since the $M$ 
		buyers matched by the optimal offline are contained within the highest 
		$n^{2/3}$ buyers, the ratio just from buyers remains the same as 
		before.
		
		From the sellers side, the online keeps the 
		highest $n - n^{2/3}$ sellers, while the offline keeps at most $n$, for 
		a ratio at most $1 - 1/n^{1/3}$.
	\end{enumerate}
	Combining both cases, the ratio is asymptotically at most:
	\begin{equation}
	1 - O\left(\frac{\log n}{n^{1/3}}\right).
	\end{equation}
	Note that the choice of $n^{2/3}$ to separate the two cases is optimal.
\end{proof}

The next step is to design an online algorithm without knowing $p$ or 
$|M(S,B)|$ beforehand. The algorithm is as follows:
\begin{enumerate}
	\item Record the first $8n^{2/3}\log n$ agents and calculate their 
	median $p'$. Buy from all sellers during this sampling phase.
	\item After the sampling starts the trading phase:
	\begin{enumerate}
		\item Buy from seller $s$ if $s \le p'$.
		\item Sell to buyer $b$ is an item is available and $b \ge p'$.
	\end{enumerate}
\end{enumerate}

\noindent
For the analysis of this algorithm, we first need a concentration result on 
the sample median $p'$.
\begin{lemma}\label{lemma:median_concetration}
	Let $X = \{1,\ldots, 2n\}$ and select $8n^{2/3}\log n$ elements from 
	$X$ 
	\emph{without} replacement. Then, their sample median $M$ satisfies:
	\begin{equation}
	\Pr[|M-n| \ge n^{2/3}] \le O\left(\frac{1}{n}\right).
	\end{equation}
\end{lemma}
\begin{proof}
	We have that:
	\begin{equation*}
	\Pr[M \ge n + n^{2/3}] 
	= \Pr[\text{more than } 4n^{2/3}\log n \text{ elements picked no less than } n 
	+n^{2/3}]
	\end{equation*}
	Since we are sampling without replacement, this is equivalent to 
	selecting $8n^{2/3}\log n$ elements uniformly at random from $X'$ 
	containing $n + n^{2/3}$ 0's and $n - n^{2/3}$ 1's and 
	having their 
	sum be greater than $4n^{2/3}\log n$. Using $X'$, $\epsilon = 
	n^{2/3}/(n-n^{2/3})$ and taking $8n^{2/3}\log n$ samples 
	in Lemma \ref{lemma:stronger-concentration}, we have:
	\begin{align*}
	\Pr[M &\ge n +n^{2/3}] = \Pr[Y \ge (1+\epsilon) \expect{Y}]
	= \Pr[Y \ge n^{2/3}] \\
	&\le \exp\left(-2\left(\frac{n^{2/3}}{(n-n^{2/3})}\right)^2 (n-n^{2/3}) 
	\frac{8n^{2/3}\log n(n-n^{2/3})}{4n^2}\right)\\
	&\le O\left(\frac{1}{n}\right).
	\end{align*}
	By symmetry, the same holds for $\Pr[M \le n -n^{2/3}]$: just 
	reverse the ordering of the agents.
\end{proof}
This shows that our sample median $p'$ might have at most $n^{2/3}$ 
agents more on one side compared to the true 
median $p$. However, this loss is negligible asymptotically, as these agents are 
a uniformly random subset of the $S\cup B$. We now show 
that buying from sellers during the sampling phase, before considering any 
buyers, can only increase the number of trades in the next phase.
\begin{lemma}\label{lemma:move-to-front}
	Let $\sigma$ be a sequence containing $n$ buyers and $n$ 
	sellers. Move an arbitrary seller the beginning of the sequence to obtain  
	$s\sigma'$. Then we have:
	\begin{equation*}
	|M(s\sigma')| \ge |M(\sigma)|.
	\end{equation*}
\end{lemma}
\begin{proof}
	Let $b'$ be the first buyer not to receive an item in $\sigma$. Clearly, if 
	$b'$ 
	doesn't exist then the number of items sold in both 
	cases is $n$. Assume we 
	sell the item bough from $s$ only if it is the last item 
	left. Then, it is sold to $b'$: otherwise $b'$ would not be the first buyer 
	not to be sold an item is $\sigma$. There are two cases:
	\begin{enumerate}
		\item If $s$ appears in $\sigma$ before $b'$: both sequences 
		continue identically as we have no items in stock after $b'$.
		\item If $s$ appears after $b'$ in $\sigma$: there is one fewer seller in 
		$s\sigma'$ after $b'$, since $s$ was moved to the front. However, this 
		can result in at one lost sale.
	\end{enumerate}
\end{proof}
Actually, we have shown that moving sellers to the beginning can only increase 
trades, which is slightly more powerful. We are now ready to state one of the 
main results of this section.
\begin{theorem}
	This algorithm is 
	$\left(1-\tilde{O}\left(\frac{1}{n^{1/3}}\right)\right)$-competitive 
	for welfare.
\end{theorem}
\begin{proof}
	As before, let $M = |M(S,B)|$ be
	the size of the optimal offline matching.  
	The following analysis assumes that the event of 
	Lemma~\ref{lemma:median_concetration} did not occur and $p$ and 
	$p'$ split the agents in two sets, differing by at most $n^{2/3}$. Given 
	this, we analyse the algorithm in three steps. First show that we never 
	buy too many items from highly valued sellers, therefore we keep most 
	of the sellers' contribution to the final welfare. Then we show that we 
	always match a high proportion of the valuable buyers by considering 
	two cases: if there are few such buyers then they are matched to the 
	sellers we obtained during the sampling phase, otherwise we have 
	enough sellers below $p'$ to match them to.
	
	We introduce some notation useful to the analysis: let $W$ be the set 
	containing the top $n-n^{2/3}$ highest valued agents. Then let $S_W, 
	B_W$ be the number of sellers and buyers respectively in $W$ and 
	$S_W',B_W'$ be how many of them appeared after the sampling phase.
	valued agents.
	To show the competitiveness of our algorithm, it suffices to find the 
	fraction of $W$ that is achieved at the end of the sequence: being 
	$(1-\tilde{O}(1/n^{1/3}))$-competitive against the top $n-n^{2/3}$ 
	agents implies a ratio of
	\begin{equation*}
	\left(1-\tilde{O}(1/n^{1/3})\right)\cdot \frac{n - n^{2/3}}{n} = 
	1-\tilde{O}(1/n^{2/3})
	\end{equation*}
	against all $n$ agents above the median and therefore the optimal offline.
	
	\paragraph{\textbf{We first show that we never lose too much 
			welfare by buying from sellers, both in the sampling and trading phase.}}
	Given $p'$, the only occasion on which a seller in $W$ is bought is 
	if he is among the first $8n^{2/3}\log n$ sellers. This event is clearly 
	independent from the condition on $p'$, meaning in expectation we keep 
	\begin{equation}\label{eqn:many_sellers}
	\expect{S_W'} = S_W\left(1 - \frac{8n^{2/3}\log n}{n}\right) = 
	S_W\left(1 - \frac{8\log n}{n^{1/3}}\right)
	\end{equation}
	highly valued sellers. Therefore, enough of the sellers' original value is 
	kept. The rest of the analysis will only focus \emph{only} the proportion of 
	buyers in $W$ who get an item. For the number of items $I_S$ bought during 
	the sampling phase, the following holds by 
	Lemma~\ref{lemma:stronger-concentration}:
	\begin{equation}
	\Pr\left[I_S \le (1 - \frac{1}{2})4n^{2/3} \log n\right] \le 
	e^{-2\frac{1}{4} 8n^{2/3}\log n \frac{n^2}{4n^2}} \le 
	e^{ -n^{2/3}},
	\end{equation}
	as there are $n$ out of $2n$ agents are sellers and we sample $8n^{2/3}\log 
	n$ of them. Therefore, we enter the trading phase with an excess of at least 
	$2n^{2/3}\log n$ items with high probability.\\

	To analyse the number of buyers in $W$ matched, we consider two cases.
	\paragraph{$\boldsymbol{B_W \le n^{2/3} \log n}$:}
	In this case there are few valuable buyers and all we need to show is that the 
	excess of items bought during sampling is enough to trade with most of them.
	We first need to find $\expect{B_W'}$, which is slightly more complicated, 
	since we have conditioned on $p'$ approximating the median.
	Given $p'$, at least $4n^{2/3}\log n$ agents were above the median 
	value 
	during the sampling phase. Note that all of the agents in $W$ are above 
	the median. Therefore, any of the agents in the upper $4n^{2/3}\log n$ 
	half of the sampling phase could be replaced by a buyer in $W$. At worst, 
	$4n^{2/3}\log n$ agents from $W$ are in the sampling phase, which means 
	that in a random permutation, we have:
	\begin{equation*}
	n^{2/3} \log n \ge \expect{B_W'} \ge B_W\left(1 - \frac{4n^{2/3}\log n}{n - 
		n^{2/3}}\right).
	\end{equation*}
	We might also consider up to $n^{2/3}$ extra buyers, if $p'$ 
	underestimated $p$.
	However, given that $I_S \ge 2n^{2/3} \log n$ with high probability, 
	every buyer in 
	$B_W'$ will be matched with an item, giving the claimed competitive 
	ratio for this case.
	
	\paragraph{$\boldsymbol{B_W > n^{2/3} \log n}$:} Let $k 
	\ge B_W$ be 
	the number of trades the optimal offline algorithm would perform.
	Since the median might be underestimated, the number of sellers we 
	consider is at least $k - n^{2/3}$ and buyers at most $k + n^{2/3}$. We 
	show that, with the help of the extra items we bought during sampling, 
	we have more items than buyers in total, with high probability. Let 
	$S_{p'},B_{p'}$ the number of sellers and buyers below and above $p'$ 
	\emph{after} the sampling phase.
	By by Lemma~\ref{lemma:stronger-concentration} we expect to find 
	\begin{align}
	\Pr[S_{p'} \le (1&-\sqrt{\frac{\log k}{k}})(k-n^{2/3})
	\frac{2n-8n^{2/3}\log n}{2n}] \\
	&\le 
	\exp\left(-2 \frac{\log k}{k} (k - n^{2/3})\frac{(2n-8n^{2/3}\log 
		n)^2}{4n^2}\right) \le O\left(\frac{1}{k}\right),
	\end{align}
	by sampling $2n-8n^{2/3}\log n$ out of $2n$ with $k - n^2/3$ important 
	elements. Similarly we have
	\begin{equation}\label{eq:extra-buyers}
	\Pr[B_{p'} \ge (1+\sqrt{\frac{\log k}{k}})(k+n^{2/3})
	\frac{2n-8n^{2/3}\log n}{2n}] \le O\left(\frac{1}{k}\right).
	\end{equation}
	It is important to note that these quantities are almost equal, other than a 
	$n^{2/3}$ factor which is insignificant compared to $k$.
	Then, with high probability (well $O(1/k)$, \emph{relatively high}):
	\begin{align}
	B_{p'} - S_{p'} &\le (1+\sqrt{\frac{\log k}{k}})(k+n^{2/3}) - 
	(1-\sqrt{\frac{\log k}{k}})(k-n^{2/3})\\
	&\le 2(\sqrt{k \log k}+n^{2/3})\\
	&\le 3(\sqrt{n \log n}),
	\end{align}
	given than $k \le n$. Since we bought at least $2n^{2/3}\log 
	n$ items during the sampling phase, the \emph{total} number of items 
	bought is higher than the total number of buyers considered for $n$ 
	large enough. Also, by Lemma~\ref{lemma:move-to-front} having these 
	items ready before encounter buyers is beneficial.
	
	Therefore, we get a lower bound on the number of buyers in $B_{p'}$ 
	that actually acquire an item using Lemma~\ref{lemma:matching-size}. The 
	number of items sold in expectation is at least:
	\begin{equation}\label{ineq:match-size}
	%\mathcal{M} \ge \frac{B_{p'} -1}{B_{p'}}\left(B_{p'} - \sqrt{2 B_{p'} \log 
	%	B_{p'}}\right).
	\mathcal{M} \ge B_{p'} - \sqrt{8 B_{p'} \log B_{p'}}.
	\end{equation}
	However, we are interested only in the fraction of buyers in $B_W$ who 
	acquired an item. The algorithm does not differentiate between any 
	buyer above $p'$, the sequence is uniformly random and all buyers in 
	$B_W$ are contained within the top $k+n^{2/3}$ buyers. By  
	lower bounding $B_{p'}$ with
	Lemma~\ref{lemma:stronger-concentration}:
	\begin{equation}
	\Pr[B_{p'} \le (1-\sqrt{\frac{\log k}{k}})(k+n^{2/3})
	\frac{2n-2n^{2/3}\log n}{2n}] \le O\left(\frac{1}{k}\right),
	\end{equation}
	and using \eqref{ineq:match-size}, the fraction of buyers in $B_W$ 
	matched is at least:
	\begin{equation}
	\frac{\mathcal{M}}{k+n^{2/3}}
	\ge 1 - \tilde{O}\left(\frac{1}{n^{2/3}}\right),
	\end{equation}
	with probability $1 - O(1/k)$, which is asymptotically high as $k \ge n^{2/3} 
	\log n$.
\end{proof}

\section{Gain from Trade}

Compared to the welfare, the gain from trade is a more challenging 
objective. The main reason is that even for large $n$, the actual trades that 
maximise the GFT can be very few and quite well hidden. Moreover, buying 
from a single seller and being unable to sell could completely shatter the 
GFT, while it could have very little effect on the welfare.

First of all, the setting has to be slightly changed. We give the online 
algorithm one extra, free item at the beginning to ensure that at least one 
buyer can acquire an item, even when the initial sampling has been 
inconclusive. For fairness, the offline algorithm is also provided with this 
starting item. We show that this modification is absolutely necessary to 
study this setting under competitive analysis.
\begin{theorem}
	Starting with no items, there exist $S,B$ such that  the competitive ratio 
	for the GFT is arbitrarily high.
\end{theorem}
\begin{proof}
	Consider two different valuations. The first has $s_1 = c > 0$ and 
	$b_1 = c+\epsilon$. In the second has $\hat{s}_1 = c, \hat{b}_1 = c - \epsilon, 
	\hat{s}_2  < \hat{b}_1 - \epsilon'$. We tweak the value of the buyer so that the 
	trade from instance one no longer increases welfare, but add one extra seller 
	to keep the optimal GFT positive.
	
	Let $p = \Pr[ s_1 \in T_S ~|~ \sigma^{-1}(b_1) > \sigma^{-1}(s_1)  ]$ be 
	the 
	probability of the online algorithm buying from $s_1$, conditioned on 
	$b_1$ arriving later. This must be $p > 0$,
	otherwise his expected gain from trade will be 0, compared to the 
	$\epsilon/2$ 
	generated by the offline. 
	
	However, in the second instance the algorithm should buy from 
	$\hat{s}_2$ 
	instead of $\hat{s}_1$. But, if $\hat{s}_1$ appears first, the first 
	algorithm should buy from him too, as the information received so far is 
	the same:
	\begin{align*}
	\hat{p} &= \Pr[ \hat{s}_1 \in T_S~|~\sigma^{-1}(\hat{b}_1) > 
	\sigma^{-1}(\hat{s}_1)  \land \sigma^{-1}(\hat{s}_2) > 
	\sigma^{-1}(\hat{s}_1)] \\
	&\ge p \Pr[\sigma^{-1}(\hat{s}_2) > 
	\sigma^{-1}(\hat{s}_1)] \\
	&= \frac{p}{2} > 0.
	\end{align*}
	So the online algorithm has a positive chance of buying the item from 
	the 
	wrong seller. Assuming in all other cases maximum gain from trade is 
	extracted, 
	we have:
	\begin{equation}
	GFT_A(S,B) \le \hat{p} (-\epsilon) + 
	(1-\hat{p})(\epsilon' - 
	\epsilon).
	\end{equation}
	Since $\hat{p}$ is independent of $\epsilon, \epsilon'$, we can set $$\epsilon' 
	= \frac{\epsilon}{1-\hat{p}}$$ which leads to $GFT_A(S,B) = 0$ whereas the 
	offline has $\epsilon/2$.
	
	In any case, no online algorithm can perform well in both instances.
\end{proof}

To avoid the previous pitfall, we assume the intermediary starts with one 
item. Roughly, the algorithms starts by estimating the total volume of 
trades in an optimal matching by observing the first segment of the 
sequence. Using this information, two prices $\hat{p},\hat{q}$ are 
computed, to be offered to agents in the second part. This being an ordinal 
mechanism, the goal is to maximise the number of trades \emph{and leave 
	no item unsold}. During the trading phase we are also much more 
	conservative: 
at most one item is kept in stock and we stop buying items well before the end 
of the sequence, to make sure that there are enough buyers left to sell 
everything.
The online algorithm 
$A(c,\epsilon, N)$ contains parameters 
whose values will be specified later.\\

\begin{algorithm}[H]
	\KwIn{A sequence $\sigma$ of length $2n$, appearing online.}
	\KwOut{A matching between buyers and sellers.}
	
	With probability $\frac{1}{2}$ ignore sellers and sell the item as in the 
	normal secretary, otherwise continue \;
	
	Split the sequence into two segments such that $\sigma = 
	\sigma_1 \sigma_2$, with $|\sigma_1| = c \cdot 2n$\;
	Let $S_1, B_1$ denote the sellers and buyers of $\sigma_1$\;
	Calculate the welfare maximising matching $M(S_1,B_1)$\;
	\If{$|M(S_1,B_1)| \le N$}
	{ 
		Sell the item to the highest buyer as in the normal secretary problem 
		and stop\;
	}
	
	Set $\hat{p},\hat{q}$ which only keep $(1-\epsilon)\cdot c 
	\cdot |M(S_1,B_1)|$ many matched pairs\;
	$i \leftarrow c\cdot 2n$\;
	$k \leftarrow \emptyset$\;
	$M \leftarrow \emptyset$\;
	\tcc{For the first half of $\sigma_2$, buy and sell items, keeping at most 
		one in stock}
	\While{$i \le c\cdot 2n+ (1-c)\cdot2n/2$} 
	{
		\If{$\sigma(i)$ is a seller, $k = \emptyset $ and $\sigma(i) \le 
			\hat{q}$}
		{
			$k \leftarrow \sigma(i)$\;
		}
		\If{$\sigma(i)$ is a buyer, $k \neq \emptyset$ and $\sigma(i) \ge 
			\hat{p}$}
		{
			Sell to $\sigma(i)$\;
			$k \leftarrow \emptyset$\;
		}
		$i \leftarrow i + 1$\;
	}
	For the second half of $\sigma_2$, just try to sell the last remaining 
	item, if any\;
\end{algorithm}

The idea is to use the first part of the sequence to estimate the matching 
$M(S,B)$. If a large (in terms of pairs) GFT maximising matching is 
observed, it is 
likely that a proportionate fraction of it will be contained in the second half. 
In that case, sellers and buyers are matched in non overlapping pairs to 
avoid buying too many items. However, if the observed matching is too small, 
then the algorithm defaults to selling only the starting item, as it is very likely that 
$\sigma_2$ will not contain enough buyers for anything more.

Before moving on to the analysis of the algorithm, we need a simple lemma 
on the structure of GFT maximising matchings, to explain the prices set.
\begin{lemma} \label{lemma:matchings}
	For any $S,B$ and $S_1 \subseteq S,B_1 \subseteq B$:
	\begin{enumerate}
		\item \label{lemma:prices} $m(S,B)$ can be obtained by setting 
		two threshold 
		prices $p,q$ and trading with buyers above and sellers below them.
		\item \label{lemma:subset} Choosing  $\hat{p} > p$ and $\hat{q} < q$ 
		such that 
		$|M(S,B,\hat{q},\hat{p})| \ge \alpha |M(S,B)|$ for $\alpha < 1$ yields 
		$m(S,B,\hat{q},\hat{p}) \ge \alpha m(S,B)$.
		\item $|M(S,B)| \ge |M(S_1,B_1)|$ and $m(S,B) \ge m(S_1,B_1)$.
	\end{enumerate} 
\end{lemma}
\begin{proof}
	For Property \ref{lemma:prices}, assume $s < b < \hat{s} < \hat{b}$, 
	such that 
	$s,b,\hat{s},\hat{b} 
	\in M(S,B)$. But, instead of two matches we can just match $\hat{b}$ to 
	$s$ instead: $\hat{b} - s > b - s+ \hat{b} - \hat{s}$, thus any such 
	pair of matched agents cannot be part of $M(S,B)$. Setting $q = 
	\max \sset{ s \in M(S,B)}$ and $p = \min \sset{ b
		\in M(S,B)}$ we have $q < p$ and the result follows. This is essentially 
	the same observation as using the median price to trade, but using two 
	different prices for robustness, as we will see later.\\
	
	Property \ref{lemma:subset} follows because $M(S,B,\hat{q},\hat{p})$ 
	contains the $\alpha$ highest value pairs for $M(S,B)$.
	Property 3 is straightforward.
\end{proof}

\begin{theorem}
	$A(c=0.3, \epsilon = 0.2758, N = 114)$ is $O(1)$-competitive for the 
	gain from trade.
\end{theorem}
\begin{proof}
	Let $z = |M(S,B)|$.  We bound the gain from trade for the case 
	where $\sigma_1$,
	$\sigma_2$ contain their analogous proportion of $M(S,B)$ and show 
	that the losses are insignificant otherwise.
	In particular, let 
	\begin{align*}
	&f(c,\epsilon,z)\\
	&= \prob{\frac{M(S,B) \cap M(S_1,B_1)}{M(S,B)} \ge 
		c(1-\epsilon) \land \frac{M(S,B) \cap M(S_2,B_2)}{M(S,B)} \ge 
		(1-c)(1-\epsilon)}
	\end{align*}
	be the \emph{well mixed} probability, where an 
	$\epsilon$-approximate chunk of the matching appears in both parts. 
	The two 
	events are not independent.
	To bound $f(c,\epsilon,z)$, it suffices to study the distribution of 
	$S_M = \sset{s \in M(S,B)}$ and $B_M = \sset{b \in 
		M(S,B)}$, the sets of agents comprising the optimal matching. By 
	Lemma~\ref{lemma:matchings}, we know that any seller in $S_M$ can 
	be matched to any buyer in $B_M$. Since we only care about the 
	\emph{size} of the matching in $\sigma_1$ and $\sigma_2$, not its 
	actual value, we can rewrite $f(c,\epsilon,z)$ as:
	\begin{align}\label{eqn: well-mixed}
	f(c,\epsilon,z) = \Pr\Bigg[ &\frac{|S_M \cap S_1|}{|S_M|} \ge 
	c(1-\epsilon) \land \frac{|B_M \cap B_1|}{|B_M|} \ge 
	c(1-\epsilon) \land \\
	&\frac{|S_M \cap S_2|}{|S_M|} \ge 
	(1-c)(1-\epsilon) \land \frac{|B_M \cap B_2|}{|B_M|} \ge 
	(1-c)(1-\epsilon) \Bigg],
	\end{align}
	which is easier to handle. 
	
	It is useful to think the input as 
	being created in two steps: first the \emph{volume} of agents in 
	$S_1,B_1, 
	S_2, B_2$ is chosen and \emph{afterwards} their exact values are 
	randomly 
	assigned. As such, a lower bound on the fraction of the size of the online 
	to the offline matching 
	provides the same bound on the gain from trade. We begin by 
	bounding 
	$f(c,\epsilon,z)$.
	
	\begin{lemma} \label{lem:good_case}
		The probability the matching is well-mixed is $$f(c,\epsilon,z) \ge 
		1 
		- 2(e^{-2 \epsilon^2 z  c^2} + 
		e^{-2\epsilon^2 z  (1-c)^2 })$$ 
	\end{lemma}
	\begin{proof}
		Continuing from \eqref{eqn: well-mixed} we have: 
		\begin{align}
		f(c,\epsilon,z) &=
		\Pr\Bigg[ \frac{|S_M \cap S_1|}{|S_M|} \ge 
		c(1-\epsilon) \land \frac{|B_M \cap B_1|}{|B_M|} \ge 
		c(1-\epsilon) \land \nonumber\\
		&\frac{|S_M \cap S_2|}{|S_M|} \ge 
		(1-c)(1-\epsilon) \land \frac{|B_M \cap B_2|}{|B_M|} \ge 
		(1-c)(1-\epsilon) \Bigg] \nonumber\\ 
		&\ge 1 - \prob{\frac{|S_M \cap S_1|}{|S_M|} \le 
			c(1-\epsilon)} - \prob{\frac{|B_M \cap B_1|}{|B_M|} \le 
			c(1-\epsilon)} - \nonumber\\
		&\prob{\frac{|S_M \cap S_2|}{|S_M|} \le 
			(1-c)(1-\epsilon)}  - \prob{\frac{|B_M \cap B_2|}{|B_M|} \le 
			(1-c)(1-\epsilon)}\label{good_case:union-bound}\\
		&\ge 
		1 - 2(e^{-2 \epsilon^2 z  c^2} + 
		e^{-2\epsilon^2 z  (1-c)^2 }),
		\label{good_case:corollary}
		\end{align}
		where \eqref{good_case:union-bound} follows by taking the 
		complement and a union bound and \eqref{good_case:corollary} by 
		applying \Cref{lemma:stronger-concentration} individually 
		for each event.
	\end{proof}
	
	Let $p$ and $q$ be the prices achieving the matching 
	$M(S,B)$, by \Cref{lemma:matchings}. We need to show that the 
	prices $\hat{p},\hat{q}$ 
	computed achieve a constant approximation of $m(S_2,B_2)$. 
	Since $M(S,B)$ is well mixed and by using 
	\Cref{lemma:matchings} we have that:
	\begin{equation}\label{eq: matching_sizes}
	|M(S,B)| \ge |M(S_1,B_1)| \ge |M(S_1,B_1,q,p)| \ge (1-\epsilon)\cdot c 
	\cdot |M(S,B)|,
	\end{equation}
	where the second inequality holds since $M(S_1,B_1)$ is a gain from 
	trade 
	maximising matching and the third because at least a 
	$(1-\epsilon)\cdot c$ fraction of $M(S,B)$ appeared in $\sigma_1$. 
	In particular, we have that $M(S_1,B_1,q,p) \subseteq M(S_1,B_1)$ is 
	the 
	highest value part of $M(S_1,B_1)$ and 
	$M(S_1,B_1, \hat{q}, \hat{p}) \subseteq M(S_1,B_1,q,p)$, thus $\hat{q} 
	\le 
	q$ and $\hat{p} \ge p$ leading to:
	\begin{equation}\label{ineq:at_best_total}
	|M(S_1,B_1,\hat{q},\hat{p})| \ge (1-\epsilon)^2 c^2 |M(S,B)|
	\end{equation}
	by \cref{eq: matching_sizes}. Therefore, the prices $\hat{p},\hat{q}$ 
	computed find a relatively large \emph{subset} of $M(S,B)$.  We now need 
	to find just how many of 
	the trades in 
	$M(S_2,B_2,\hat{p},\hat{q})$ are achieved by our algorithm. Let 
	$\hat{S}_2 = \sset{s \fwh{s \in 
			S_2 \land s < \hat{q}} }$ and $\hat{B}_2 = \sset{b 
		\fwh{b \in B_2 \land b > 
			\hat{p}} }$. We need a high probability guarantee on the size of 
	$\hat{S}_2$ and $\hat{B}_2$.
	
	\begin{lemma}\label{lemma:size-in-s2}
		Assuming the matching is well mixed: $$\prob{|\hat{S}_2| \ge 
			((1-c)(1-\epsilon) - \frac{1}{2}) 
			|S_M|} \ge 1 - 2^{-c^2(1-\epsilon)^2 |S_M|}.$$ 
	\end{lemma}
	\begin{proof}
		In the well-mixed case, we have that
		\begin{equation}\label{ineq:mixing1}
		|S_2 \cap S_M| \ge (1-c)(1-\epsilon)|S_M| \text{ and } |S_1 \cap S_M| 
		\ge c(1-\epsilon)|S_M| 
		\end{equation}
		which leads to
		\begin{equation}
		|S_1 \cap S_M| \le (1 - (1-c)(1-\epsilon))|S_M|.
		\end{equation}
		To get a lower bound on the size, we have:
		\begin{align}
		\prob{ |\hat{S}_2| \ge \frac{|S_M|}{2} - |S_1 \cap S_M| } 
		&\ge \prob{ |\hat{S}_2| \ge ((1-c)(1-\epsilon) - \frac{1}{2}) |S_M|} 
		\label{ineq:mixing2}\\
		&\ge \prob{ \hat{q} \ge \text{median}(S_M)}\label{ineq:median}\\
		&\ge 1 - 2^{-c^2(1-\epsilon)^2 |S_M|}\label{ineq:independent}
		\end{align}
		where \cref{ineq:mixing2} follows from \cref{ineq:mixing1}. 
		\cref{ineq:median} follows since if $\hat{q}$ is greater than the median, 
		then at worst case all elements from $S_1 \cap S_M$ are less than 
		$\hat{q}$, which still leaves plenty of sellers in $S_2$. 
		\cref{ineq:independent} follows since draws are not actually independent, 
		but this works in the inequality's favour. From \cref{ineq:at_best_total} we 
		know $\hat{q}$ is greater than at least a $c^2(1-\epsilon)^2$ 
		fraction of sellers. Since the `bad case' is choosing all sellers below the 
		median, this happens with higher probability if each draw is \emph{with} 
		rather than \emph{without} replacement, leading to the result.
	\end{proof}

	\noindent Clearly, \Cref{lemma:size-in-s2} holds for buyers as well. The 
	proof is almost identical, keeping in mind that buyers are ordered the 
	opposite way. \\
	
	\noindent At this point we have a clear indication of how many sellers and 
	buyers the prices $\hat{p},\hat{q}$ cover in the second part of the 
	sequence. Since this is an ordinal mechanism, we want to maximise the 
	number of trades \emph{provided no item is left unsold}. There are no 
	a priori guarantees on the welfare increase of each trade, even a single 
	unsold item ruins our gain from trade guarantees, in the worst case.
	
	\begin{lemma}\label{lemma:online_matching}
		Let $A = |M(S_2,B_2,\hat{q},\hat{p})|$ and $B = |S_2| + |B_2| - A$. 
		Then, 
		the probability that no item is left unsold is at least $1-2^{-A}$. 
		Moreover, the expected number of trades in this case is at least:
		\begin{equation}
		\frac{\frac{A+B/2-1}{2A+B-1} \cdot \frac{A}{2} - \frac{A}{2^A}}{1 - 
			2^{-A}} \approx \frac{|M(S_2,B_2,\hat{q},\hat{p})|}{4}.
		\end{equation}
	\end{lemma}
	\begin{proof}
		We begin by calculating the probability of having an unsold item, which 
		is 
		easy: it is at most as much as the probability of not encountering a buyer 
		within the last $(1-c)\cdot 2n/2$ agents. Using a similar argument as 
		\Cref{lemma:size-in-s2}, this probability is at most $2^{-A}$.\\
		
		\noindent We now need to calculate the expected number of trades. Let 
		$X_i$ be a random variable indicating that an item was \emph{sold} to 
		the $i$-th agent. We have:
		\begin{align*}
		\prob{X_i = 1} &= \prob{\text{ previous 
				transaction was buying } \land X_i \text{ is a buyer }} \\
		&= \frac{A}{2A}\cdot \frac{A}{2A+B-1} = \frac{1}{2} \cdot 
		\frac{A}{2A+B-1},
		\end{align*}
		since the previous transaction being buying from a seller occurs with 
		probability $\frac{A}{2A}$ as there are $A$ sellers in 
		$M(S_2,B_2,\hat{p},\hat{q})$ for $2A$ total agents and the sequence is 
		shuffled. The second fraction has $2A+B-1$ for the denominator, taking 
		into account that one seller has already been used.\\
		
		\noindent
		By linearity of expectation, the total number of trades $X$ is (note that we 
		only consider the first half of $\sigma_2$, where we both buy \emph{and} 
		sell):
		\begin{align*}
		\expect{X} &\ge \expect{\sum_{i=2}^{(2A+B)/2} X_i} = (A+B/2 - 1) 
		\expect{X_i} \\
		&= (A+B/2 - 1)\frac{1}{2} \cdot \frac{A}{2A+B-1}\\
		&= \frac{A+B/2-1}{2A+B-1} \cdot \frac{A}{2} 
		\end{align*}\\
		
		\noindent
		We can use this to calculate the expected number of trades in the case 
		where nothing is left unsold:
		\begin{align*}
		\expect{X | \text{No items unsold}} &= \frac{\expect{X} - \expect{X | 
				\text{ Unsold items}}\prob{\text{Unsold Items}}}{\prob{\text{No unsold 
					items}}} \\
		&\ge \frac{\expect{X} - \frac{A}{2^{-A}}}{1 - 2^{-A}}
		\end{align*}
	\end{proof}

	Everything is now in place to provide a lower bound on the gain from trade 
	of the 
	matching calculated by the algorithm. Assuming $z = |M(S,B)|$, we can 
	compose 
	\Cref{lem:good_case}, \Cref{lemma:size-in-s2} and 
	\Cref{lemma:online_matching} to show that with probability at least
	\begin{equation}
	J(c,\epsilon,z) = f(c,\epsilon,z) \cdot (1 - 2^{-c^2(1-\epsilon)^2z})\cdot 
	(1-2^{-((1-c)(1-\epsilon) - \frac{1}{2})z}),
	\end{equation}
	the matching has size at least
	\begin{equation}
	\frac{((1-c)(1-\epsilon) - \frac{1}{2})z}{4}.
	\end{equation}
	The matching is not a uniformly random subset of $M(S,B)$, but it is skewed 
	to contain higher value trades since $\hat{p} > p$ and $\hat{q} < q$. 
	Taking into account that we run a simple secretary algorithm with 
	probability $1/2$ and assuming we lose the highest valued seller $s^\star$ 
	in our matching when the agents are not well mixed (we can only have one 
	unsold item) the GFT is at least:
	\begin{equation}
	\frac{1}{2e}b^1 + \frac{J(c,\epsilon,z)}{2}\cdot \frac{((1-c)(1-\epsilon) - 
		\frac{1}{2})m(S,B)}{4} - \frac{1 - J(c,\epsilon,z)}{2}s^\star
	\end{equation}
	whereas the offline GFT is at most
	\begin{equation}
	m(S,B) + b^1.
	\end{equation}
	To upper bound the competitive ratio $\rho$ we analyse three different cases:
	\begin{enumerate}
		\item If $z < N$:\\
		In this case the algorithm would never detect a sufficiently sized 
		matching and would always run a simple secretary algorithm. Note this is 
		possible, as $c = 0.3 \le 1/e$ required for the secretary.
		\begin{equation}
		\rho \ge \frac{\frac{1}{e} b^1}{m(S,B) + b^1} \ge \frac{1}{e}\cdot 
		\frac{b^1}{(z+1)b^1} = \frac{1}{e(z+1)}
		\end{equation}
		\item If $N \le z < N\frac{1}{c(1-\epsilon)}$. In the well-mixed case, the 
		online algorithm will not detect a matching and fall back to secretary. 
		Therefore, the competitive ratio is:
		\begin{equation}
		\rho \ge \frac{
			\frac{1}{2e}b^1 + \frac{1}{2}(f(c,\epsilon,z) 
			\frac{1}{e}b^1 - (1-f(c,\epsilon,z)) s^\star
		}
		{
			m(S,B)+b^1
		}
		\ge \frac{
			1-e + (1+e)f(c,\epsilon,z)
		}
		{
			2e(z+1)
		}.
		\end{equation}
		given that $c < 1/e$ and the sampling phase for the secretary continues.
		\item $z \ge N\frac{1}{c(1-\epsilon)}$. Now in the well mixed case 
		a large enough matching is found. We have:
		\begin{equation}
		\frac{
			\frac{1}{2e}b^1 - \frac{1 - J(c,\epsilon,z)}{2}s^\star
		}
		{
			b^1
		}
		\ge \frac{1}{2e} - (1 - J(c,\epsilon,z)),
		\end{equation}
		and
		\begin{equation}
		\frac{
			\frac{J(c,\epsilon,z)}{2}\cdot \frac{((1-c)(1-\epsilon) - 
				\frac{1}{2})m(S,B)}{4}
		}
		{
			m(S,B)
		}
		\ge \frac{
			J(c,\epsilon,z)\cdot ((1-c)(1-\epsilon) - 
			\frac{1}{2})m(S,B)
		}{8}.
		\end{equation}
		Therefore, the competitive ratio is:
		\begin{align*}
		\rho &\ge \frac{	\frac{1}{2e}b^1 + \frac{J(c,\epsilon,z)}{2}\cdot 
			\frac{((1-c)(1-\epsilon) - 
				\frac{1}{2})m(S,B)}{4} - \frac{1 - J(c,\epsilon,z)}{2}s^\star}
		{m(S,B) + b^1}\\
		&\ge \min\left\{\frac{1}{2e} - (1 - J(c,\epsilon,z)),
		\frac{J(c,\epsilon,z)\cdot ((1-c)(1-\epsilon) - 
			\frac{1}{2})m(S,B)
		}{8}\right\}.
		\end{align*}
	\end{enumerate}
	Therefore, $c,\epsilon$ and $N$ are selected to maximise the minimum 
	amongst all cases of $z$, which is picked by the adversary. 
	Computationally, we find that setting $c = 0.3, \epsilon = 0.2758$ and $N 
	= 114$ yields $\rho \ge 1/1434$.
\end{proof}
If we are given that $|M(S,B)|$ will be large, then this algorithm can be 
adapted to have greatly improved competitive ratio. In particular, setting 
$c=\epsilon=0.01$ achieves $\rho \ge 1/17$ as $|M(S,B)| \rightarrow 
\infty$.

\bibliography{intermediation}
\bibliographystyle{plain}

\end{document}